\documentclass[sigconf,nonacm]{acmart}
\setcopyright{none}

\usepackage{amsmath}
\usepackage{booktabs}
\usepackage{float}
\usepackage{graphicx}
\usepackage[nameinlink,noabbrev]{cleveref}
\newtheorem{definition}{Definition}
\newtheorem{theorem}{Theorem}
\newtheorem{corollary}{Corollary}
\newtheorem{proposition}{Proposition}
\newtheorem{lemma}{Lemma}
\newcommand{\Adv}{\mathsf{Adv}}
\newcommand{\LRp}{LR+}\newcommand{\Dep}{\mathsf{Dep}}\newcommand{\Sem}{\mathsf{Sem}}

\AtBeginDocument{%
  
}

\title{Beyond Object Authentication: Context-Closed Post-Quantum Authentication for the WebPKI}

\author{Anis Bkakria}
\affiliation{%
  \institution{IRT SystemX}
  \city{Palaiseau}
  \country{France}%
}

\acmConference[Preprint]{arXiv preprint}{August 2026}{}
\acmYear{2026}



\begin{document}
\sloppy
\raggedbottom
\begin{abstract}
Post-quantum migration increases WebPKI authentication cost, but authenticating a compressed certificate object does not by itself preserve the mutable authorization context under which a relying party accepts it.  We formalize \emph{context closure}: the authenticated projection accepted by a verifier must determine the selected authorization semantics it claims, relative to declared source contracts and event-coverage witnesses.  We instantiate this idea with \LRp, a two-plane post-quantum construction that authenticates mutable CA-context state in an update plane while the warm path carries only state-local dependency references selected by explicit profile negotiation.

In a pinned CCADB reconstruction, we obtain 44,912 path/view contexts and 16,858 physical CA lineages across Apple, Chrome, Microsoft, and Mozilla views.  The core compiler yields $m_{50}=6$, $m_{95}=16$, and $m_{\max}=18$ typed dependencies.  A warm LR+ selector therefore costs 296, 776, and 872 bytes at median, p95, and maximum, compared with 3,842, 5,932, and 6,350 bytes for a one-signature stateless bundle carrying the same dependency vector.  The retained all-view closure state is 16.15 MB, and per-view lifecycle crossovers range from 19.60 to 50.41 median-path warm authentications/day under the stated checkpoint and update model.  The implementation and evaluation artifact are available at \url{https://github.com/nserser/LR-WebPKI}.
\end{abstract}
\maketitle

\section{Introduction}
\label{sec:introduction}

Post-quantum migration turns certificate authentication into a bandwidth problem.  Standardized signatures such as ML-DSA are intended to withstand quantum adversaries, but their public keys and signatures are much larger than the classical signatures around which today's WebPKI was engineered~\cite{fips204}.  A natural response is to stop paying the full authentication cost on every connection.  Certificates can be represented relative to authenticated tree state, signatures can be amortized across many objects, and relying parties can retain authenticated state between connections.  Merkle Tree Certificates (MTCs) and Merkle Tree Ladder (MTL) mode are examples of this broader direction~\cite{mtc05,mtl09}.  Our concern is not whether such compression works.  It is what the compressed representation still has to authenticate.

The problem appears when the same authentic object participates in several certification contexts.  Suppose a CA key is cross-signed under two parents and a lower certificate remains byte-for-byte unchanged.  A relying party selects one lineage.  Later, an authority fact attached to that lineage---a status generation, a parent authorization, a root-program state, or a platform distrust rule---changes independently of the certificate.  The certificate and its local inclusion evidence can still verify perfectly, even though the selected context should now be rejected.  Replacing the authenticator on the unchanged object with ML-DSA, a larger Merkle commitment, an aggregate proof, or a succinct argument cannot help if the changing authority fact is absent from the authenticated statement.

We call the missing condition \emph{context closure}.  An accepted transcript, together with the relying party's retained monotone state, should determine every independently mutable fact needed to evaluate the selected path.  The first theorem isolates the omission principle: if two admissible worlds give the authentication compiler and verifier the same complete information but require different selected decisions, stronger cryptography over that information cannot close the gap.  The second theorem characterizes the authority evidence needed to avoid it, under an explicit witness taxonomy.  A decision-relevant event must be reflected by current authenticated state ($\mathsf F$), excluded by authenticated complete coverage ($\mathsf C$), or deliberately neutralized by an authenticated bounded grant ($\mathsf L$).  Each source supplies its own safe horizon, and only the horizons of dependencies selected by the current context contribute to that context's deadline.

LR+ resolves three design constraints that are easy to conflate.  First, it separates full end-entity validation from retained CA-context terminals.  A full selected validation instance is
\[
P_{\rm full}=(C_{\rm EE},X_0,X_1,\ldots,R,V),
\]
whereas LR+ retains the CA authorization lineage
\[
\bar P=(X_0,X_1,\ldots,R,V).
\]
The end-entity certificate remains in the baseline evidence; retained terminal topology scales with CA authorization contexts, not the live population of TLS leaf certificates.  Second, the warm path does not carry or match a global state identifier.  Global state identifiers remain on the update/checkpoint plane, while warm origin authentication is a state-local selector interpreted against the verifier's installed authenticated state.  Since different clients may use different trust-store views and selected CA lineages, LR+ makes selector negotiation explicit: clients can either advertise a coarse closure-profile identifier or receive a small catalog of profile-indexed selectors.  Third, an untyped depth proxy is replaced by a semantically justified compiler: for the evaluated core profile, LR+ emits root-profile and root-authorization dependencies plus selected-edge and selected-context status records for each selected CA transition.

The setting is not confined to a synthetic cross-signing example.  We pin the structural evaluation to the public CCADB archive release \texttt{v1.20260828.235636}.  The artifact downloads the full V5 metadata report, 33 yearly certificate-PEM reports for 1994--2026, and root-trust input; records are normalized into a multigraph and path/view rows are reconstructed from that rebuilt corpus.  The resulting reference graph contains 10,234 records/nodes, 24,653 edges, 365 roots, 228,249 graph paths, and 17,051 active TLS paths before per-view expansion.  Under four configured views, the compiler sees 44,912 path/view contexts representing 16,858 distinct physical CA-certificate sequences.  These figures do not claim that every multi-parent record is vulnerable, nor that every browser path can be reconstructed by CCADB alone.  They show that public WebPKI data contains substantial structure in which local certificate identity and selected authorization context are different things.

LR+ is our concrete realization of context closure.  It separates a relatively infrequent authority-update plane from the warm-authentication path used on repeated connections.  Mutable authority records, including source-specific horizons, are committed separately from stable CA-context topology.  A publisher authenticates canonical state transitions with ML-DSA-44.  Clients install consecutive deltas, signed range catch-up deltas, or authenticated monotone checkpoints.  Warm connections carry stable dependency references into state the client has already authenticated; the verifier derives the retained CA-context terminal locally from the selected lineage and those references.  The third theorem reduces an incorrect accepted context to update-signature forgery, hash/commitment binding failure, ordinary baseline-validation failure, adaptive selected-source failure, or trusted-time failure.

The executable profile gives byte-exact closure-layer overheads.  Over the pinned CCADB reconstruction, the typed compiler yields global $m_{50}=6$, $m_{95}=16$, and $m_{\max}=18$.  With the state-local warm wire codec, LR+ adds 296 bytes of warm selected-context closure at the median, 776 bytes at p95, and 872 bytes at the maximum.  In profile-hint mode, the transmitted warm object is only $(\tau,\rho_1,\ldots,\rho_m)$: the terminal identity and check value are re-derived by the verifier from the exact selected CA lineage and the transmitted references, then checked against the current compiler vector committed in topology.  A strong stateless Signed Path Bundle (SPB), carrying the same typed dependency vector and one ML-DSA-44 signature, costs 3,842, 5,932, and 6,350 bytes at the same corpus points.  The retained all-view reference state is 16,150,580 bytes; the signed full checkpoint is 16,160,724 bytes.  At 30-day checkpoint amortization and 16 changed typed records/day, the median-path crossover against SPB is 19.60 warm authentications/day for Chrome, 37.98 for Mozilla, 47.02 for Apple, and 50.41 for Microsoft.

\paragraph{Contributions.}
We make four contributions.
\begin{itemize}
  \item \textbf{Context closure and event coverage.} We formalize the gap between authenticating an object and preserving the authorization semantics of a selected context.  We prove the projection-omission separation and characterize decision-relevant mutable events, under an explicit witness taxonomy, through reflected state ($\mathsf F$), authenticated complete coverage ($\mathsf C$), or bounded deliberate irrelevance ($\mathsf L$).
  \item \textbf{LR+, a post-quantum realization.} LR+ combines CA-context terminals, path-local horizons, stable topology, authenticated delta/checkpoint recovery, and state-local warm selectors.  Its security argument separates primitive assumptions from ordinary validation, source semantics, and trusted time.
  \item \textbf{Typed WebPKI dependency compilation.} We replace an untyped depth placeholder by an executable CA-context compiler that emits root-profile, root-authorization, selected-edge, and selected-context status dependencies, each with typed provenance, source contract, horizon rule, evidence digest, and correlation well-formedness metadata.  We also report a policy-augmented sensitivity profile to show how richer declared semantics break the deterministic depth-to-count relation.
  \item \textbf{Reproducible evaluation and lifecycle model.} We run the pipeline from a pinned CCADB archive release, regenerate raw path/view traces, report global and per-view communication/state results, and model bootstrap/update/warm lifecycle cost relative to a strong stateless baseline.
\end{itemize}

Section~\ref{sec:wrong-context} starts from the two-world failure.  Section~\ref{sec:context-closure} develops the semantic model and the full-path/CA-context split.  We then present LR+, prove its selected-context security, examine the WebPKI reconstruction and source contracts, evaluate the executable profile, and discuss deployment limits and related work.

\section{Authenticated Object, Wrong Context}
\label{sec:wrong-context}

Certificate compression changes where authentication evidence is carried.  A conventional path sends several signatures and status-bearing objects with the connection; a compressed design can authenticate many objects against shared tree state, cached landmarks, or a separately distributed authenticated frontier.  MTCs, for example, authenticate CA-issued entries through issuance logs and allow landmark-relative certificates when the relying party already holds sufficiently recent trusted landmark state~\cite{mtc05}.  This is attractive for post-quantum bandwidth, but it exposes a boundary: an authentic certificate, log entry, or local tree state need not imply that the \emph{selected certification context} containing it remains authorized.

\begin{figure}[t]
  \centering
  \includegraphics[width=\columnwidth]{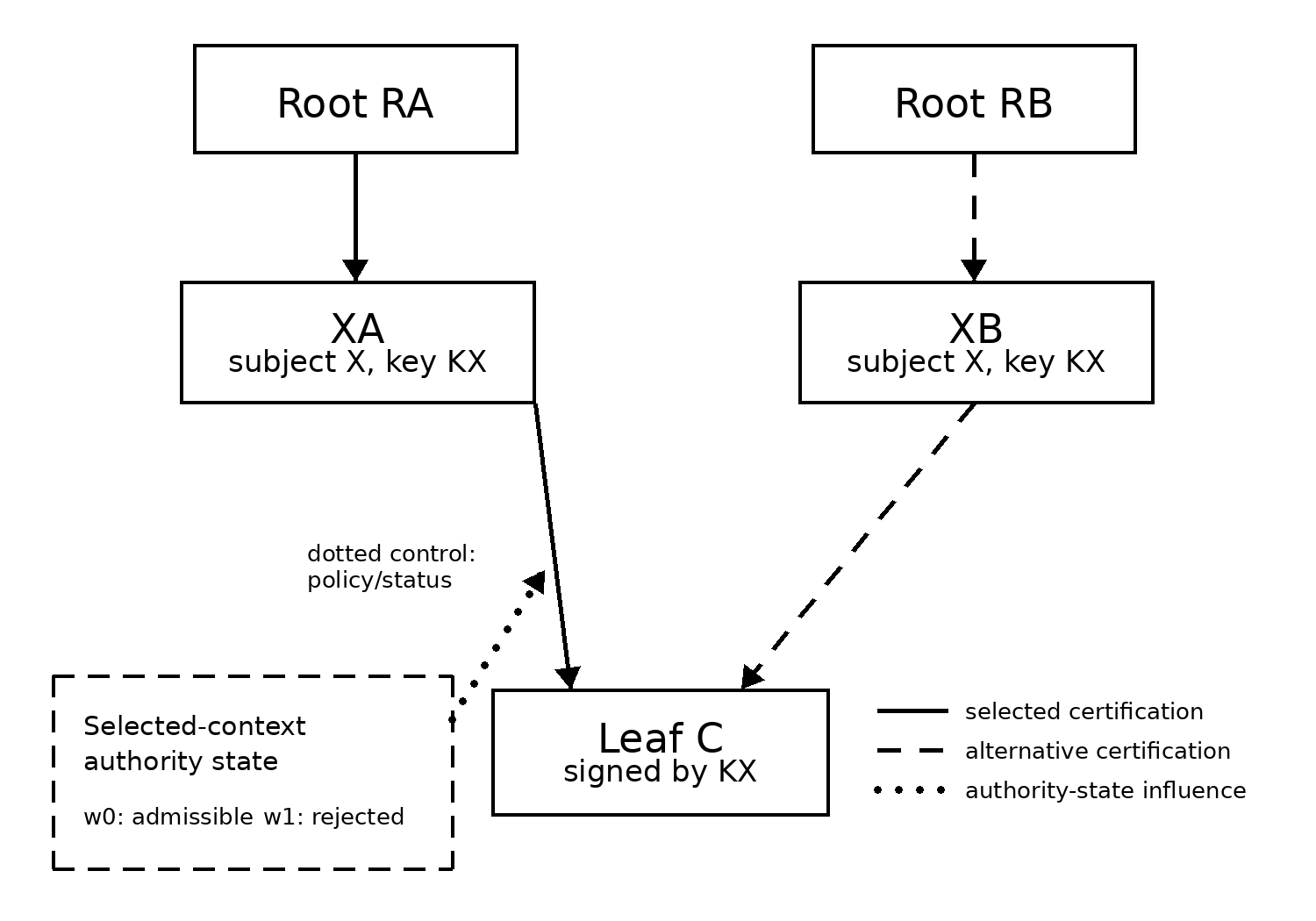}
  \caption{Two certification contexts authenticate the same CA subject and key $K_X$, and therefore the same lower object $C$. Solid arrows mark the selected certification path through $X_A$; dashed arrows mark the alternative path through $X_B$; the dotted control arrow marks mutable policy/status authority state rather than another certification edge. The selected-context authority state is admissible in world $w_0$ and no longer admissible in world $w_1$, even though the certificate objects remain unchanged.}
  \Description{Schematic with two roots and two cross-signed certificates for the same subject and public key converging on the same lower certificate. Solid arrows show the selected left-hand path and dashed arrows show an alternative path. A dashed authority-state box labels world w0 as admissible and world w1 as not admissible, with an arrow to the selected context.}
  \label{fig:wrong-context}
\end{figure}

\paragraph{A concrete two-world execution.}
Consider \Cref{fig:wrong-context}.  Certificates $X_A$ and $X_B$ bind the same CA subject and public key $K_X$ under different parent contexts.  A lower certificate $C$ signed by $K_X$ is therefore compatible with both lineages.  A full TLS validation instance may include an end-entity certificate below $C$; for the semantic point here, $C$ denotes the lowest object whose issuer context is selected by the retained CA lineage.  Suppose the relying party selects the left context.

At time $t_0$, the selected $R_A\!\rightarrow X_A$ context is admissible.  The relying party obtains a compressed transcript containing authentic evidence for the certificate objects and whatever local issuance-tree state the underlying compressor requires.  It may even cache that evidence exactly as intended by the compression scheme.  At time $t_1$, the bytes of $C$ have not changed; neither has its signature under $K_X$.  What changes is an independently mutable authority fact attached to the selected $R_A\!\rightarrow X_A$ context.  In world $w_0$ the selected context remains admissible; in world $w_1$ it is no longer admissible.  Any transcript that authenticates only the unchanged objects collapses the two worlds.

The example is not an attack on cross-signing itself.  Cross-signing is useful and common.  The point is that selected authorization is a relation among object authenticity, parent context, root-program state, status and distrust state, and time.  A compressed representation that authenticates only the object component can be cryptographically strong and still semantically incomplete.

\paragraph{CA-context terminals.}
The design uses this example in a leaf-independent way.  Let
\[
P_{\rm full}=(C_{\rm EE},X_0,X_1,\ldots,R,V)
\]
be the full validation instance under trust-store view $V$.  The ordinary baseline evidence $B_P$ authenticates the end-entity object, its selected issuer $X_0$, and certificate-local validity decisions.  LR+ retains a terminal for the CA authorization lineage
\[
\bar P=(X_0,X_1,\ldots,R,V),
\]
not for every end-entity certificate.  This is the right scale for the CCADB evidence used in the evaluation: CCADB reconstructs CA root/intermediate structure, while live leaf issuance remains in the baseline or in a separate leaf-status mechanism.

\section{Context-Closed Authentication}
\label{sec:context-closure}

The failure in \Cref{sec:wrong-context} arises before any cryptographic primitive fails.  The verifier can hold authentic evidence and still be missing information needed to decide whether the \emph{selected} context is authorized.  We therefore first ask what an accepted transcript must determine about the world; only later do we choose how to bind that information cryptographically.

\subsection{Full paths, CA contexts, and lifting}

The LR+ profile separates end-entity validation from retained CA-context closure.  A full selected validation instance is
\[
P_{\rm full}=(C_{\rm EE},X_0,X_1,\ldots,R,V),
\]
where $C_{\rm EE}$ is the end-entity object, $X_0$ is the selected issuing CA, $R$ is the root, and $V$ is the trust-store view.  The retained CA-context lineage is
\[
\bar P=(X_0,X_1,\ldots,R,V).
\]
The baseline evidence $B_P$ includes the certificate objects and signatures used by ordinary validation, certificate-local extensions such as Basic Constraints, Key Usage, Name Constraints, certificate policies, SAN/name bindings, and validity intervals, plus any leaf-local status evidence assigned to the baseline.  LR+ closes the mutable CA authorization context in which that baseline object is accepted.

\begin{proposition}[Full-path lifting]
\label{prop:lifting}
For the declared CA-context profile, suppose path building and baseline static validation return
\[
  \mathsf{BaseValidate}(P_{\rm full})=(\mathsf{accept},\bar P,\mathsf{leafState}),
\]
where $\bar P$ is the exact selected CA-context lineage, not merely the issuing certificate $X_0$.  If LR+ context-closes this retained lineage, then baseline validation and LR+ together determine the selected full-path authorization semantics for every mutable decision assigned to either interface.
\end{proposition}

\begin{proof}
Take two admissible worlds with equal baseline evidence and equal installed LR+ state for the same selected $\bar P$.  Equality of the baseline fixes the end-entity object, the selected concrete CA lineage output by path building, and all leaf-local decisions assigned to baseline validation.  Equality of the installed LR+ state fixes the effective root, selected-edge, status, policy, and grant values consumed by that lineage.  The declared profile factors selected full-path semantics through exactly these two interfaces.  Hence the selected full-path semantics agree.  The statement would be false if the baseline exposed only $X_0$: the same issuing key may occur under several parent/root contexts.
\end{proof}

\subsection{Worlds and sufficient context}

Fix a selected context $P$, a verification time $t$, and an admissible world $w$.  The world contains authenticated local histories, authority and policy/status events, selected certification context, and relying-party progress state relevant to validation.  Let $B_P(w)$ denote the evidence supplied by the underlying certificate-authentication or compression substrate, and let $\Sem_P(w,t)$ denote the selected semantics the new layer promises to preserve.

\begin{definition}[Context sufficiency]
A context value $X_P(w)$ is \emph{sufficient} for $(P,t)$ relative to $B_P$ if, for every pair of admissible worlds $w_0,w_1$,
\[
\begin{aligned}
  B_P(w_0)&=B_P(w_1),\\[-1mm]
  X_P(w_0)&=X_P(w_1)
\end{aligned}
\quad\Longrightarrow\quad
\Sem_P(w_0,t)=\Sem_P(w_1,t).
\]
A compressed-authentication protocol is \emph{context-closed} when every accepted transcript, together with the relying party's trusted monotone state, is computationally bound to such a sufficient context and to the temporal rule under which that context remains usable.
\end{definition}

The definition is representation-independent.  $X_P$ may be an explicit vector, a Merkle root, a vector commitment, or mostly retained client state.  What matters is the induced partition on admissible worlds: worlds requiring different selected decisions may not collapse to the same authenticated state.

\begin{theorem}[Context omission]
\label{thm:omission}
Let $\pi(w)$ be the complete world-dependent input consumed by an authentication compiler, and let $\rho(w)$ be the trusted relying-party state available before the candidate authentication is processed.  If there exist admissible worlds $w_0,w_1$ such that
\[
  \pi(w_0)=\pi(w_1),\qquad \rho(w_0)=\rho(w_1),
\]
but $\Sem_P(w_0,t)\ne\Sem_P(w_1,t)$, then no compiler whose output distribution depends on the world only through $\pi(w)$, together with a verifier whose additional trusted pre-state is $\rho(w)$, can provide context closure for $(P,t)$.
\end{theorem}

\begin{proof}
Equality of $\pi$ gives the compiler the same world-dependent input in both worlds, so its signatures, commitments, Merkle proofs, aggregate proofs, or succinct arguments have the same distribution.  Equality of $\rho$ also fixes the verifier's trusted information before that transcript is processed.  The complete information available to the accepted decision is therefore identically distributed, even though the required selected semantics differ.
\end{proof}

The theorem is intentionally modest: it is a separation lemma.  Its value is to make the next question unavoidable: which authority events must an accepted context determine, exclude, or neutralize?

\subsection{Dependency completeness and typed compilation}

For each mutable selected dependency $S_i$, let $h_i$ be its admissible authority history and let
\[
A_i(h_i,t)
\]
be the effective policy-relevant value contributed by $S_i$ at time $t$.  The value is evaluated after applying any valid grant whose semantics the policy recognizes.  Thus an $\mathsf L$ witness may allow an underlying authority fact to change while keeping the effective value unchanged until the grant expires; an early-revocation channel, when present, is modeled as another dependency.

\begin{definition}[Dependency completeness]
\label{def:dependency-completeness}
A declared dependency set $\Dep(P)=\{S_1,\ldots,S_m\}$ with abstractions $A_1,\ldots,A_m$ is complete for $P$ if, for all admissible worlds $w,w'$ and verification times $t$ in the supported domain,
\[
\begin{gathered}
 B_P(w)=B_P(w'),\\[-1mm]
 A_i(h_i(w),t)=A_i(h_i(w'),t)\quad\forall i
\end{gathered}
\Rightarrow
\Sem_P(w,t)=\Sem_P(w',t).
\]
\end{definition}

This is the explicit no-hidden-dependency condition.  It has three distinct roles in the paper.  First, it is the abstract condition needed by context closure.  Second, Section~\ref{sec:compiler} proves an exact compiler-completeness statement for a declared core CA-context semantics, denoted $\Sem_{\rm core}$.  Third, coverage of complete live browser behavior is a deployment claim and is not asserted here: if a browser or root program makes another mutable fact decision-relevant, a deployment must add another stream or include the fact inside an existing composite dependency before the theorem is applied to that richer profile.

\subsection{F/C/L as determine, exclude, neutralize}

An accepted descriptor $\delta_i$ induces the compatible-history set
\[
\Gamma_i(\delta_i,t)=\{h_i: h_i\text{ is admissible and consistent with }\delta_i\text{ at }t\}.
\]
It is semantically safe through $t$ when all compatible histories agree on the effective value consumed by policy:
\[
  \forall h,h'\in\Gamma_i(\delta_i,t),\quad A_i(h,t)=A_i(h',t).
\]
The labels $\mathsf F$, $\mathsf C$, and $\mathsf L$ are extensional ways of establishing this condition under the declared witness taxonomy.  They are not protocol names; they classify what the descriptor does to the compatible-history space.

\paragraph{$\mathsf F$: reflected state.}
Authenticated evidence determines the current effective branch.  A nonce-bound OCSP response can be such evidence at acquisition time, but a timestamp alone is not enough to guarantee forward freshness.

\paragraph{$\mathsf C$: excluded event by complete coverage.}
Authenticated coverage establishes that every event in a declared region is represented, so an unrepresented effective event cannot exist there.  A useful descriptor exports an effective-coverage watermark, not merely a publication cutoff.

\paragraph{$\mathsf L$: event neutralized by grant.}
An authenticated grant deliberately makes the underlying event irrelevant to the selected decision through an expiry.  If the grant can be revoked earlier, that revocation channel is itself a dependency.

These modes are semantic operators: determine, exclude, and neutralize.  A composite descriptor may use different witnesses for different event classes.  Its safe horizon is the minimum of the constituent horizons, and the path horizon is the minimum over dependencies selected by the path:
\[
D_{P,e}=\min_{S_i\in\Dep(P)}D_{i,e}.
\]

\begin{theorem}[Selected-path F/C/L event-coverage characterization]
\label{thm:fcl}
Fix $(P,t)$, assume dependency completeness, and assume a globally event-toggle-closed authority model for selected independently mutable dependencies.  Under the declared witness taxonomy, if a decision-relevant event class of a selected dependency is neither determined by a valid $\mathsf F$ witness, excluded by a valid $\mathsf C$ witness, nor neutralized by a valid $\mathsf L$ witness through $t$, the accepted context is not sufficient for $(P,t)$.  Conversely, if every event class that can affect selected dependency values is covered through $t$ by such a witness, the resulting descriptors are bound, and policy/time evaluation is deterministic and sound, then the accepted dependency vector is context-sufficient.
\end{theorem}

\begin{proof}
For necessity, global toggle closure supplies admissible worlds with equal baseline evidence, selected context, relying-party state, and all other dependency histories, while the uncovered event branch differs.  Because the declared taxonomy is extensional, an uncovered event is precisely one whose compatible histories are not forced into the same effective value by determination, exclusion, or neutralization.  Decision relevance then makes the selected semantics differ, so Theorem~\ref{thm:omission} applies.  For sufficiency, suppose two accepted worlds with equal bound descriptors have different selected semantics.  Dependency completeness implies some selected effective value differs.  A valid $\mathsf F$ witness fixes that value, a valid $\mathsf C$ witness excludes the hidden alternative, and a valid $\mathsf L$ witness makes the underlying variation policy-irrelevant before expiry.  Each case contradicts the assumed effective-value difference.
\end{proof}

\begin{lemma}[Sufficient condition for global toggles]
\label{lem:toggle}
Suppose that, after fixing baseline evidence, verifier pre-state, and all non-selected dependency histories, each declared independent event class admits a conditionally rectangular set of locally admissible branches, and the profile does not split correlated authority facts across independent dependencies.  Then local event independence implies the global toggle condition required by Theorem~\ref{thm:fcl}.
\end{lemma}

\begin{proof}
Rectangularity supplies one global admissible world for each locally admissible branch while fixed components remain unchanged.  Correlated facts are grouped into the same effective dependency.  The construction therefore never asks for an impossible mixed history.
\end{proof}

\paragraph{Withholding versus rollback.}
A monotone epoch stops a client that has already advanced from reinstalling old authenticated state.  It does not make a suppressed successor visible.  LR+ handles withholding separately: an installed state remains usable only while the selected-path clock rule certifies $t\le D_{P,e}$; beyond that point, suppression becomes fail-closed unavailability.  Because the horizon is path-local, a short-lived dependency in an unrelated context does not expire $P$.

\section{LR+: Context-Closed PQ Authentication}
\label{sec:lrplus}

Context closure would be unattractive if it replaced certificate signatures with a large context proof on every connection.  LR+ avoids that outcome by separating two timescales.  Authority state changes on an update plane; repeated connections carry compact selectors into state the relying party has already authenticated.

The design relies on three separations.  First, retained terminals represent CA-context lineages $\bar P$, while end-entity evidence remains in $B_P$.  Second, global state identifiers authenticate update/checkpoint transitions but are absent from the warm origin message.  Third, dependency counts come from a typed compiler that assigns semantic records to profile, root, edge, and status decisions.

\subsection{Dependency state and stable topology}

At closure epoch $e$, the mutable dependency map $M_e$ contains typed records
\[
 m_{s,e}=(s,\mathsf{type},c_s,g_{s,e},v_{s,e},\mu_{s,e},D_{s,e},\eta_{s,e}).
\]
Here $s$ is a stable stream identifier; $\mathsf{type}$ identifies the semantic class; $c_s$ is the authority context; $g_{s,e}$ is a generation; $v_{s,e}$ is the effective validation value; $\mu_{s,e}$ is the source contract; $D_{s,e}$ is the horizon; and $\eta_{s,e}$ binds evidence metadata.  Let $R_e=\mathsf{Root}(M_e)$.

Canonical identities are domain-separated.  For every concrete certificate $X$, let
\[
  \mathsf{cid}(X)=H(\textsf{CERT-ID},\operatorname{DER}(X)),
\]
where $\operatorname{DER}(X)$ is the canonical certificate encoding.  Thus two cross-signed certificates with the same subject and public key but different issuers have different certificate identifiers.  For a retained lineage $\bar P=(X_0,\ldots,R,V)$, define
\[
  \mathsf{pid}_{\bar P}=H(\textsf{PATH-ID},
  \operatorname{enc}(V,\mathsf{cid}(X_0),\ldots,\mathsf{cid}(R))).
\]
A typed stream descriptor $s=(\mathsf{type},c_s)$ has
\[
\begin{aligned}
  \mathsf{strid}(s)&=H(\textsf{STREAM-ID},
      \operatorname{enc}(\mathsf{type},c_s)),\\
  \rho_s&=H(\textsf{REF},\operatorname{enc}(\mathsf{strid}(s))).
\end{aligned}
\]
Stable references therefore name streams, not epochs or current values.

Topology is held separately.  Let $\mathsf{Compile}_v(\bar P)$ be the compiler output under topology/profile version $v$, and let
\[
  \mathsf{Refs}_v(\bar P)=\mathsf{CanonicalRefs}(\mathsf{Compile}_v(\bar P))
      =(\rho_1^*,\ldots,\rho_m^*)
\]
be its canonical reference vector.  Define
\[
 d_{\bar P}^*=H(\textsf{ORDERED-REFS},
        \operatorname{enc}(\rho_1^*,\ldots,\rho_m^*))
\]
and
\[
 c_{\bar P}^*=H(\textsf{TERM-CHECK},\mathsf{pid}_{\bar P},m,d_{\bar P}^*).
\]
The current terminal topology $T_v$ is keyed by $\mathsf{pid}_{\bar P}$ and stores exactly one current terminal statement
\[
  T_v[\mathsf{pid}_{\bar P}]=(m,d_{\bar P}^*,c_{\bar P}^*)
\]
for each supported retained lineage.  This is the terminal/compiler consistency invariant: a terminal is not merely some authenticated vector for a path identity; it is the vector currently emitted by the compiler for that lineage.  Retained topology no longer duplicates full mutable records for every path, and record-only updates do not rewrite terminals.  If a compiler/profile change changes $\mathsf{Refs}_v(\bar P)$, it is a topology update and must advance $v$.

A full-state client stores
\[
\mathcal C_{e,v}=(\mathsf{cfg},pk_{\rm upd},e,v,R_e,Q_v,M_e,T_v,\mathsf{sid}_{e,v}),
\]
where
\[
\mathsf{sid}_{e,v}=H(\textsf{STATE},\mathsf{cfg},e,v,R_e,Q_v)
\]
is an update-plane identifier.  There is no global closure deadline.  For a selected context, the verifier derives $D_{\bar P,e}$ from the authenticated records selected by $\bar P$.

\subsection{Delta and checkpoint installation}

LR+ supports two signed delta forms on the update plane.  A \emph{consecutive delta} from $(e-1,v)$ to $(e,v')$ contains configuration, predecessor and successor state identifiers, canonical mutable-map changes, canonical topology changes, and a ML-DSA-44 signature.  The client installs it only if the configuration and key are pinned, the predecessor equals its installed state, the epoch is consecutive, topology versioning is monotone, the signature verifies, and local application recomputes exactly the successor roots and state identifier.

A \emph{range catch-up delta} is an optional signed transition from an installed state $(e,v)$ to a later state $(e',v')$, with $e'>e$.  It contains both endpoint state identifiers and roots, the final values/generations/horizons for every typed record whose effective state differs across the range, the topology endpoint, and a domain-separated publisher signature.  The client accepts it only from its exact installed predecessor; all patched record generations must advance monotonically; applying the payload to the installed maps must recompute the signed successor roots and state identifier; and topology versioning must satisfy the same monotonicity rule as consecutive deltas.  Thus a range delta is not a freshness shortcut and not a warm-path proof: it is a compact update-plane catch-up object whose accepted endpoint is equivalent to installing the corresponding valid consecutive transition sequence.

\begin{proposition}[Range endpoint equivalence]
\label{prop:range-delta}
Assume a pinned update key, collision-resistant state commitments, canonical delta encodings, and the topology-version invariant.  If a range delta from installed state $\mathcal C_{e,v}$ to endpoint $\mathcal C_{e',v'}$ is accepted, then the resulting local state is identical to the endpoint obtained by applying any valid consecutive sequence from $e$ to $e'$ with the same final typed-record values and topology endpoint.  Conversely, any missing, stale, non-monotone, or tampered changed record causes the recomputed endpoint root or state identifier check to fail, except with a hash-binding failure.
\end{proposition}

A checkpoint authenticates a complete current snapshot.  It is domain-separated from deltas and range deltas, includes the configuration, epoch, topology version, roots, state identifier, and snapshot digest, and is accepted only if it is monotone relative to the installed state.  The topology invariant is explicit:
\[
 v'=v\Rightarrow Q_{v'}=Q_v.
\]
Thus any topology-root change increments the topology version.  Consecutive deltas give exact predecessor chaining; range deltas give signed endpoint catch-up from a known installed state; checkpoint mode gives authenticated monotone recovery.  A newer endpoint remains usable for warm authentication only while selected records' horizons justify it.

\subsection{State-local warm authentication}
\label{sec:warm-wire}

The warm origin message does not match the client's global state identifier, and it also does not transmit a terminal key or terminal check value.  Its wire object is
\[
\mathsf{WireWarm}=(\tau,\rho_1,\ldots,\rho_m),
\]
where each $\rho_i$ is a 48-byte stable dependency reference.  The 8-byte header is parsed as
\[
  \tau=(\mathsf{type}:1,\mathsf{version}:1,\mathsf{flags}:2,m:4)
\]
in network byte order.  Parsing rejects unless the type is \textsf{LRPLUS-WARM}, the version is supported, reserved flags are zero, $m$ is within the configured maximum, and the encoded byte string has exactly $8+48m$ bytes with no trailing or truncated reference material.  Thus
\[
 B_{\rm LR+,warm}(m)=8+48m.
\]

Baseline path building and static validation return the exact selected CA-context lineage $\bar P$ for the full validation instance $P_{\rm full}$.  From $\bar P$ and the transmitted references the verifier computes
\[
 d_{\rm wire}=H(\textsf{ORDERED-REFS},\operatorname{enc}(\rho_1,\ldots,\rho_m)),
\]
\[
 c_{\rm wire}=H(\textsf{TERM-CHECK},\mathsf{pid}_{\bar P},m,d_{\rm wire}).
\]
It then queries the installed topology at $\mathsf{pid}_{\bar P}$.  Acceptance requires
\[
  T_v[\mathsf{pid}_{\bar P}]=(m,d_{\rm wire},c_{\rm wire}).
\]
By the terminal/compiler consistency invariant above, this equality implies that the transmitted reference vector is exactly $\mathsf{Refs}_v(\bar P)$, except with the hash-binding error accounted for in Theorem~\ref{thm:security}.  Hence the warm wire cannot select a stale terminal vector after the compiler has added, removed, or reordered dependencies for the same lineage.

Warm verification proceeds against the verifier's installed authenticated state:
\begin{enumerate}
  \item parse $\mathsf{WireWarm}$ and reject malformed length, version, flag, or count encodings;
  \item build a candidate path and run baseline static validation, obtaining $(\mathsf{accept},\bar P,\mathsf{leafState})$;
  \item derive $\mathsf{pid}_{\bar P}$, $d_{\rm wire}$, and $c_{\rm wire}$ from $\bar P$ and $(\rho_1,\ldots,\rho_m)$;
  \item check $T_v[\mathsf{pid}_{\bar P}]=(m,d_{\rm wire},c_{\rm wire})$;
  \item resolve each reference in $M_e$ in canonical order, check source-contract modes and horizons, and evaluate the typed effective values under the declared profile.
\end{enumerate}

\begin{proposition}[Warm wire-to-compiler binding]
\label{prop:warm-wire}
Assume collision-resistant domain-separated encodings and the terminal/compiler consistency invariant for $T_v$.  If warm verification accepts $\mathsf{WireWarm}=(\tau,\rho_1,\ldots,\rho_m)$ for selected lineage $\bar P$ in installed state $(M_e,T_v)$, then
\[
  (\rho_1,\ldots,\rho_m)=\mathsf{Refs}_v(\bar P).
\]
Moreover, the transmitted warm bytes are exactly the parsed 8-byte header plus the $m$ 48-byte references; all terminal identities and check values used by the verifier are deterministic local functions of the selected lineage and this reference vector.
\end{proposition}

\begin{proof}
Parsing fixes a unique value of $m$ and a unique ordered reference vector.  The verifier computes $d_{\rm wire}$ and $c_{\rm wire}$ from that vector and the locally selected $\bar P$, then requires $T_v[\mathsf{pid}_{\bar P}]=(m,d_{\rm wire},c_{\rm wire})$.  The invariant says the unique current terminal stored at that path identifier commits to $\mathsf{Refs}_v(\bar P)$.  If the accepted wire vector differs from that compiler output, the two distinct canonical vectors produce the same ordered-reference or terminal-check digest, which is exactly the hash-binding event already charged in the security theorem.
\end{proof}

A replayed warm selector contains references, not old authority values.  The verifier evaluates it against its current installed state.  If an unrelated record changes, the same wire selector still resolves the selected records in the new state.  If a compiler/profile update changes the selected dependency vector, the topology terminal for $\mathsf{pid}_{\bar P}$ changes and old wire vectors fail.  If a selected record value changes, the selector resolves the new effective value or fails by policy or horizon checks.  If the verifier is stale because updates were withheld, residual acceptance is bounded by selected dependency horizons and source-contract failure terms.

\subsection{View and path negotiation}
\label{sec:view-negotiation}

State-local selectors are secure but not self-negotiating.  The selector is bound to the locally selected lineage
\[
  \bar P=(X_0,\ldots,R,V),
\]
so an origin that serves heterogeneous clients must choose a selector compatible with the client's closure profile and candidate path.  LR+ treats this as an explicit negotiation problem rather than requiring the origin to guess a browser's trust-store view.  This is aligned with ongoing TLS work on trust-anchor identifiers, where relying parties can convey supported trust anchors and authenticating parties can use that signal to choose a compatible certificate representation~\cite{trustanchorids,mtc05}.

Let
\[
\begin{aligned}
  \mathsf{profid}=H(&\textsf{LRPLUS-PROFILE},\mathsf{pub},V,\\
                   &\mathsf{compilerVersion},\mathsf{rootSetGeneration}).
\end{aligned}
\]
be a coarse closure-profile identifier.  It identifies a publisher, trust-store view, compiler/profile version, and root-set generation.  It intentionally excludes the client's exact installed epoch, mutable root, topology root, and update-plane state identifier.  Advertising a profile therefore supports selector choice without recreating global-state synchronization.

LR+ supports two delivery modes.  In \emph{profile-hint mode}, the client supplies \(\mathsf{profid}\) or an equivalent coarse trust-anchor/profile signal, and the origin returns one \(\mathsf{WireWarm}\) for a candidate lineage under that profile.  The server-to-client closure overhead remains
\[
  B_{\rm hint}(m)=8+48m,
\]
with any client hint accounted separately.  In \emph{catalog mode}, the origin sends a bounded set
\[
  \mathsf{Catalog}=(\tau_{\rm cat},(\mathsf{profid}_j,\mathsf{WireWarm}_j)_{j=1}^s).
\]
The catalog header is eight bytes and every profile identifier is 48 bytes, so
\[
  B_{\rm cat}=8+\sum_{j=1}^{s}\left(48+8+48m_j\right).
\]
The client selects the unique entry whose \(\mathsf{profid}\) matches its installed closure profile and then runs the ordinary warm verifier.  If no entry matches, if two entries match, or if the entry names the wrong lineage for the client's path builder, verification fails before authority values are used.

Profile-hint mode minimizes server bytes but exposes a coarse profile signal.  Catalog mode avoids a client profile hint at the cost of sending several selectors; it is useful for origins that already maintain a small set of candidate lineages for major browser/root-program views.  Neither mode exposes the client's exact installed state.  In both modes, security still comes from local derivation of \(\mathsf{pid}_{\bar P}\) and the current-terminal check in Section~\ref{sec:warm-wire}; negotiation only determines which candidate wire vector is attempted.

\subsection{Typed CA-context compiler}
\label{sec:compiler}

For a CA-context lineage $\bar P=(X_0,X_1,\ldots,R,V)$ with CA-edge depth $d$, the evaluated core compiler emits
\[
\Dep_V(\bar P)=\{S_{\rm rprof}(R,V),S_{\rm root}(R,V)\}\cup\mathcal E_{\bar P},
\]
where
\[
\mathcal E_{\bar P}=\{S_{\rm edge}(X_i,X_{i+1},V),S_{\rm cstat}(X_i,X_{i+1},V):0\le i<d\}.
\]
Hence the evaluated core dependency count is
\[
 m(d)=2+2d.
\]
In this evaluated core profile, the regular dependency shape makes the count a deterministic function of CA-edge depth.  This does not mean that depth itself is the security model, nor that all profiles must have this regular shape.  The previous untyped depth heuristic is replaced by a semantically justified compiler: the core profile instantiates a root-profile record and a root-authorization record for the selected root/view pair, and two typed mutable records for each selected CA transition.  The $S_{\rm rprof}(R,V)$ stream binds the compiler/profile version, source-contract namespace, and root-set generation used to interpret root $R$ under view $V$; it is intentionally keyed by the selected root and view, not only by the browser view.  The $S_{\rm root}(R,V)$ stream captures the root-program authorization fact itself.  Each selected-edge stream captures selected-parent authorization and policy-generation state for that transition.  Each $S_{\rm cstat}$ stream captures status or distrust facts whose source contract is scoped to the selected context $(X_i,X_{i+1},V)$.

\paragraph{Canonicalization and correlated facts.}
The notation above is deliberately scoped.  Root-profile records are root/view scoped, explaining why their all-view count equals the root-authorization count rather than the number of browser views.  In the artifact output these records keep the family label 	exttt{PROFILE}; their stream scope is nevertheless $(R,V)$, recorded through the root identifier and profile identifier columns, and should not be read as a view-only stream.  Similarly, the evaluated artifact has one context-status stream per selected-edge stream; the equality between these two counts is expected, not an accidental measurement claim.  If a deployment has a source contract proving that a status fact is independent of the selected parent context, the compiler may instead canonicalize it as a shared CA/view stream $S_{\rm castat}(X_i,V)$ and reuse it across parent edges.  The paper's headline numbers do not rely on this optional deduplication.  A compiler is correlation-well-formed, written $\mathsf{WF}_{\rm corr}(\Dep_V(\bar P))$, when authority facts that are updated only jointly, or whose admissible histories are not product-separable, are emitted inside one composite record or assigned an explicit joint-admissibility group.

\begin{proposition}[Core-profile compiler completeness]
\label{prop:compiler}
Let $\Sem_{\rm core}$ be the declared CA-context semantics whose mutable inputs are exactly the effective values of $S_{\rm rprof}(R,V)$, $S_{\rm root}(R,V)$, $S_{\rm edge}(X_i,X_{i+1},V)$, and $S_{\rm cstat}(X_i,X_{i+1},V)$, together with any composite groups required by $\mathsf{WF}_{\rm corr}$.  Assume $\mathsf{WF}_{\rm corr}(\Dep_V(\bar P))$.  If two admissible worlds agree on baseline evidence $B_P$ and on every effective dependency value emitted by the compiler for $\bar P$, then they agree on $\Sem_{\rm core}$ for the selected full path.
\end{proposition}

\begin{proof}
The statement is deliberately profile-relative.  By definition, $\Sem_{\rm core}$ reads only the root-profile, root-view, selected-edge, and context-status effective values assigned to the compiler, with correlated facts grouped as required by $\mathsf{WF}_{\rm corr}$.  Agreement on the emitted effective values therefore fixes the mutable CA-context part of $\Sem_{\rm core}$.  Together with equal baseline evidence and Proposition~\ref{prop:lifting}, the selected full-path semantics of the core profile agree.  This proposition is not a claim of complete live browser-policy reconstruction; richer browser semantics require a richer declared profile and compiler.
\end{proof}

\subsection{Comparison boundary}

LR+ is one stateful realization of context closure, not a replacement for Merkle trees or state-distribution systems.  Table~\ref{tab:semantic-scope} compares what each approach authenticates, rather than assigning a security ranking.  Independent PQ authentication can bind a full path when that path is included, but current dependency statements must then be carried again.  MTL inherits the semantics of the application messages supplied to it.  MTC natively authenticates issuance/log state relative to distributed landmarks and also includes relying-party maintained issuance-log state such as trusted subtrees and revoked serial ranges; LR+ is concerned with the broader selected-path authorization closure that composes root/profile, parent, status, distrust, policy, and bounded-grant dependencies.  Profile~S and LR+ both carry selected-context closure semantics; they differ mainly in state placement.

\begin{table*}[t]
  \caption{Semantic scope of the compared approaches. ``App.'' means application-defined; ``external'' means complementary mechanisms supply the corresponding path/status semantics.}
  \label{tab:semantic-scope}
  \centering
  \small
  \begin{tabular}{@{}p{0.17\textwidth}p{0.19\textwidth}p{0.15\textwidth}p{0.18\textwidth}p{0.20\textwidth}@{}}
    \toprule
    Approach & Authenticated unit & Client state & Context binding & Mutable authority \\
    \midrule
    Independent per-path PQ & Objects / dependency statements & None required & Optional full path & Resend current statements \\
    MTL-style~\cite{mtl09} & Message series & App.-defined & App.-defined & App.-defined \\
    MTC landmark-relative~\cite{mtc05} & Issuance / log state & Landmark / subtree & Issuance scope & Trusted subtrees / revoked ranges + complements \\
    Profile S (this paper) & Signed closure snapshot & Optional epoch cache & Ordered terminal & Per-record horizons \\
    LR+ (this paper) & Closure roots + selectors & Monotone closure state & CA-context terminal & F/C/L + $D_{\bar P,e}$ \\
    \bottomrule
  \end{tabular}
\end{table*}

\section{Security of LR+}
\label{sec:security}

The security goal is selected-context soundness.  A quantum polynomial-time adversary wins if it makes a verifier accept a warm transcript for a selected full validation instance whose ideal semantics under the declared profile reject, while the adversary has not caused one of the explicitly modeled failure events.  The selected path is adaptive: the reduction considers the first accepted winning transcript, not a union over all supported paths.

For an adversary $\mathcal A$, let $P^\star$ be the selected context in its first accepted winning transcript, if such a transcript exists.  Define the adaptive selected-source failure term
\[
 \epsilon_{\rm src}^{\rm sel}(\mathcal A)=
 \Pr\!\left[\exists S_i\in\Dep(P^\star):\mathsf{SrcBad}_i\right],
\]
where the probability is over the experiment, the source histories, and the adversary's adaptive choice of transcript.  This term is intentionally selected-context local but adversary-adaptive: it does not assume that marginal source failures remain marginal after the adversary observes state and chooses a path.

\begin{theorem}[Conditional PQ selected-context soundness]
\label{thm:security}
For the LR+ profile with ML-DSA-44 update signatures, canonical injective encodings followed by domain-separated SHA-384 commitments under quantum collision resistance, sound baseline validation, trusted time, and the declared source contracts, any QPT adversary $\mathcal A$ satisfies
\[
\begin{aligned}
 \Adv^{\rm ctx\mbox{-}sel}_{\rm LR+}(\mathcal A)
 \le{}& \Adv^{\rm pq\mbox{-}euf}_{\mathsf{Sig}}(\mathcal B_{\rm sig})
       +\Adv^{\rm qcr}_{H}(\mathcal B_H)\\
 &     +\epsilon_{\rm base}^{\rm sel}+\epsilon_{\rm clk}
       +\epsilon_{\rm src}^{\rm sel}(\mathcal A).
\end{aligned}
\]
\end{theorem}

\begin{corollary}[Path-local source instantiation]
\label{cor:adaptive-src}
Suppose each source contract satisfies the adaptive conditional guarantee
\[
\Pr[\mathsf{SrcBad}_i\mid T,\;S_i\in\Dep(P^\star)]\le \epsilon_i^{\rm src}
\]
for every admissible prior adversarial transcript $T$ and every selected dependency event in the final transcript.  Then
\[
 \epsilon_{\rm src}^{\rm sel}(\mathcal A)
 \le
 \sup_{P\in\mathcal P_{\rm cfg}}
   \sum_{S_i\in\Dep(P)}\epsilon_i^{\rm src}.
\]
In particular, if all selected streams admit the same adaptive bound $\epsilon_{\rm src}^{\max}$ and every supported context selects at most $m_{\max}$ dependencies, then $\epsilon_{\rm src}^{\rm sel}(\mathcal A)\le m_{\max}\epsilon_{\rm src}^{\max}$.
\end{corollary}

\paragraph{Proof overview.}
The update plane is authenticated independently of the warm plane.  If the adversary installs a state not produced by the pinned publisher through a consecutive update, range update, or checkpoint, it either forges an authenticated state-transition/checkpoint signature or finds a collision in a committed canonical state.  Conditioned on no such event, the verifier's installed $(M_e,T_v)$ is one of the authenticated states.  A warm wire selector has no global $\mathsf{sid}$, no terminal key, no terminal check value, and no authority values; it is interpreted only against that installed state.  The verifier derives $\bar P$ from baseline path building, parses the wire vector, and checks the terminal at $\mathsf{pid}_{\bar P}$ in its installed topology.  By Proposition~\ref{prop:warm-wire}, any accepted vector equals $\mathsf{Refs}_v(\bar P)$, the current compiler output for that selected lineage, except with hash-binding failure.  The verifier therefore resolves exactly the typed dependencies required by the installed profile.  If negotiation supplied the wrong profile or lineage, the locally derived path identifier does not match the current terminal and the warm vector is rejected before source values are evaluated.  If no selected source-contract failure occurs and trusted time is within the path horizon, every reconstructed effective dependency value equals the ideal authority value used by the declared core profile.  Core-profile compiler completeness and full-path lifting then imply the accepted selected semantics, contradicting the winning condition.

The base theorem uses $\epsilon_{\rm src}^{\rm sel}(\mathcal A)$ because the adversary may choose its final accepted context adaptively after seeing authenticated state.  Corollary~\ref{cor:adaptive-src} recovers the compact path-local summation when source contracts are stated in the stronger conditional form.  This separation prevents a marginal source-error bound for fixed paths from being used as though it automatically survived adaptive path selection.

\paragraph{Source and publisher boundary.}
The terms $\epsilon_i^{\rm src}$ and $\epsilon_{\rm src}^{\rm sel}$ are operational source-contract failure bounds.  They may arise from a protocol guarantee, deployment SLA, explicit model assumption, or measured source reliability.  A valid update signature proves what the publisher authenticated; it does not by itself prove that the upstream source was true.  A pinned publisher that deliberately signs incompatible views has not forged ML-DSA.  Transparency, witnessing, or multi-party consistency mechanisms can be composed to address equivocation, but they are outside the core game.

\paragraph{QROM scope.}
The LR+ composition proof assumes post-quantum EUF-CMA signatures and quantum collision resistance.  It does not program, reprogram, or measure random-oracle queries, so the LR+ composition argument itself needs no QROM step.  Any lower-level QROM analysis of ML-DSA or SHA-384-derived instantiations is modular to the theorem above.

\section{WebPKI Evidence and Source Contracts}
\label{sec:webpki}

The formal model is motivated by public WebPKI structure, but the evaluation is deliberately scoped.  We do not claim a live browser census, live OCSP/CRLite acquisition, or complete platform-policy reconstruction.  We use a pinned public CCADB release to rebuild CA authorization structure and then run the typed compiler over the regenerated path/view trace.

\paragraph{CCADB reconstruction.}
The pinned archive release tag is \texttt{v1.20260828.235636}, resolved to commit prefix \texttt{fbbcc8027606}.  The reproducibility manifest records the full commit and input hashes, including the V5 report, 33 yearly PEM reports covering 1994--2026, and the root-trust input.  Before rebuilding, the existing local corpus is copied to a timestamped backup.  The existing normalizer and graph builder are then run without hand-editing graph or path rows.

The rebuilt graph contains 10,234 records/nodes, 24,653 edges, 365 roots, 228,249 graph paths, and 17,051 active TLS paths before per-view expansion.  The reconstruction treats certification data as a multigraph rather than forcing a tree.  Certificates are keyed by concrete certificate identity, while Subject--SPKI grouping is retained separately to expose cross-signed identities.  Parent candidates come from explicit metadata where available and otherwise issuer/subject plus AKI/SKI compatibility; a candidate is retained only when it is CA-capable and cryptographically verifies the child signature.

Trust is evaluated separately under Apple, Chrome, Microsoft Server Authentication, and Mozilla Websites views.  The root-set evidence is materialized as one row per view/root relation: 352 roots are represented for each configured view.  The compiler emits 44,912 selected path/view contexts, representing 16,858 distinct physical CA-certificate sequences and 72,020 typed dependency records.  Identical physical lineages under different trust-store views remain distinct path/view contexts, because the selected authorization semantics include the view.

\paragraph{Source contracts.}
Event coverage cannot be inferred from a protocol name alone.  A nonce-bearing OCSP response gives a clean $\mathsf F$ example at acquisition time: the responder-signed status is tied to a requester nonce and response production time~\cite{rfc6960,rfc9654}.  This does not turn \texttt{nextUpdate} into proof that no newer revocation exists.  A forward horizon requires an explicit source-latency/effectivity contract, whose violation is an $\epsilon_i^{\rm src}$ event.

Coverage-oriented mechanisms such as CRLite illustrate $\mathsf C$.  A browser-deployed compressed revocation filter is useful only relative to an authenticated universe, scope, and coverage cut.  LR+ consumes such a coverage statement where it applies and attaches a semantic deadline; it does not infer completeness from a filter digest alone~\cite{crlite17,crlite26}.  Bounded leases illustrate $\mathsf L$ when the policy deliberately accepts a grant until expiry.  If the grant can be revoked early, that revocation channel is another dependency.

\paragraph{Deployment publisher.}
The natural publisher for the evaluated profile is a browser-vendor or root-program closure publisher.  Authority sources feed a closure compiler; the publisher signs deltas and checkpoints; browser update infrastructure distributes them; and origins send ordinary object authentication plus LR+ warm selectors.  Origins do not publish mutable authority state.  This is also the clean way to compose LR+ with MTC: MTC can authenticate certificate-issuance objects, CRLite/OCSP can handle leaf-local status when selected, and LR+ closes the mutable CA-context authorization state.

\section{Evaluation}
\label{sec:evaluation}

The evaluation asks whether the semantics above survive a concrete WebPKI-scale reconstruction.  We check the reference state machine, recompute warm communication from the typed compiler, measure retained-state/checkpoint size, and report per-view deployment costs from a pinned CCADB reconstruction.  It does not measure an end-to-end TLS handshake, browser integration, or live authority-source acquisition.

\subsection{Executable gates}

The artifact implements CA-context terminals, state-local warm selectors, typed dependencies, monotone deltas/checkpoints, view negotiation, signed range catch-up deltas, and reconstruction/evaluation scripts.  The verification target downloads the pinned CCADB input, backs up the pre-existing corpus, rebuilds a normalized WebPKI graph, reconstructs per-view path rows, evaluates communication/state/lifecycle metrics, checks outputs, and freezes a manifest.  The verification log reports 18 property tests and a full \texttt{make verify} pass.

The production compiler and SRWS codec are applied to the regenerated path/view trace.  Dependency counts are typed compiler outputs.  The compiler emits
\[
  m(d)=2+2d
\]
with root-profile, root-authorization, selected-edge, and selected-context status records.

\begin{table}[t]
  \caption{Pinned CCADB reconstruction and evaluation. Numbers are regenerated from raw path/view rows emitted by the typed compiler and state-local warm-selector codec.}
  \label{tab:eval-summary}
  \centering
  \small
  \begin{tabular}{@{}lr@{}}
    \toprule
    Metric & Result \\
    \midrule
    Pinned CCADB release & \texttt{v1.20260828.235636} \\
    Graph nodes / edges & 10,234 / 24,653 \\
    Selected path/view contexts & 44,912 \\
    Distinct physical CA lineages & 16,858 \\
    Typed $m$ median / p95 / max & 6 / 16 / 18 \\
    LR+ warm median / p95 / max & 296 / 776 / 872 B \\
    SPB warm median / p95 / max & 3,842 / 5,932 / 6,350 B \\
    Typed dependency records & 72,020 \\
    Retained state / signed checkpoint & 16,150,580 / 16,160,724 B \\
    Verification gate & 18 tests; \texttt{make verify} pass \\
    \bottomrule
  \end{tabular}
\end{table}

\subsection{Warm communication and baselines}

The state-local warm wire codec has an 8-byte type/version/count header and one 48-byte reference per selected typed dependency:
\[
  B_{\rm LR+,warm}(m)=8+48m.
\]
The global state identifier, terminal key, and terminal check value are not transmitted in the warm origin message.  The verifier derives the terminal key and check value from the selected CA-context lineage and the transmitted reference vector as described in Section~\ref{sec:warm-wire}.  SPB is a strong concrete stateless baseline: it carries the same typed dependency vector in one self-contained server-carried object and authenticates it with a single ML-DSA-44 signature.  Using 209-byte dependency records, a fixed 168-byte bundle header, and a 2,420-byte ML-DSA-44 signature gives
\[
  B_{\rm SPB}(m)=2,588+209m.
\]
Profile~S carries one signed state certificate plus selected records and membership proofs.  The independent per-dependency PQ baseline signs each typed dependency separately and includes the 209-byte dependency statement with each signature:
\[
  B_{\rm perdepPQ}(m)=m(2,420+209).
\]

\begin{center}
\small
\begin{tabular}{@{}lrrrr@{}}
\toprule
Corpus point & LR+ & SPB & Profile~S & Per-dep. PQ \\
\midrule
Median ($m=6$) & 296 B & 3,842 B & 8,993 B & 15,774 B \\
p95 ($m=16$) & 776 B & 5,932 B & 18,033 B & 42,064 B \\
Maximum ($m=18$) & 872 B & 6,350 B & 19,841 B & 47,322 B \\
\bottomrule
\end{tabular}
\end{center}

The 296-byte median follows from the SRWS codec and the typed median $m=6$.  Removing the 48-byte global warm state identifier makes room for root/profile dependencies, while the p95 is 776 bytes because the pinned reconstruction exposes deeper selected CA lineages.  The communication result is tied to the semantic dependency vector actually used by the declared profile.

\paragraph{Selector negotiation overhead.}
The numbers above are profile-hint mode, where the client supplies a coarse closure-profile identifier and the origin sends one selector.  Catalog mode is a compatibility mechanism for deployments in which the origin does not know the client's view in advance.  The artifact records catalog byte accounting as a representative sensitivity result; we do not use it to claim catalog dominance, because a deployable catalog should be evaluated over matched physical lineages for a given origin rather than arbitrary representative rows.  The paper communication claims therefore use the profile-hint mode and the per-view results below.

\subsection{Per-view reconstruction results}

\begin{table}[H]
\caption{Per-view reconstruction results from the pinned CCADB reconstruction. Checkpoint sizes are shown in MB; exact byte counts are in the artifact outputs.}
\label{tab:per-view}
\centering
\scriptsize
\begin{tabular}{@{}lrrrrr@{}}
\toprule
View & Ctx. & $m_{50}/m_{95}/m_{\max}$ & LR+ med./p95 & Ckpt. MB & Cross. \\
\midrule
Apple & 13,795 & 6/18/18 & 296/872 B & 4.798 & 47.02 \\
Chrome & 4,671 & 6/10/14 & 296/488 B & 1.881 & 19.60 \\
Microsoft & 16,710 & 8/16/18 & 392/776 B & 5.646 & 50.41 \\
Mozilla & 9,736 & 6/10/14 & 296/488 B & 3.836 & 37.98 \\
\bottomrule
\end{tabular}
\end{table}

Table~\ref{tab:per-view} is generated from \texttt{path\_view\_trace.csv} and \texttt{per\_view\_dependency\_records.csv}.  The view-specific root evidence is materialized separately in \texttt{view\_root\_sets.csv}, which contains one row per view/root relation.  The all-view aggregate remains useful for global stress testing, but the per-view state and checkpoint values are the structural deployment estimates for a single configured browser/root-store profile.  They are not traffic-weighted browser telemetry.

\subsection{Profile sensitivity and semantic heterogeneity}
\label{sec:profile-sensitivity}

The evaluated core profile has a deliberately regular shape: $m(d)=2+2d$.  This makes the headline $m$-distribution a typed transformation of reconstructed CA-edge depth.  The point of the compiler is therefore not that the core experiment already breaks the depth relation, but that each counted item is a stream with a semantic type, authority scope, source contract, horizon, evidence digest, and correlation group.

To check that the interface is not restricted to depth-derived profiles, the artifact includes a small policy-augmented sensitivity experiment.  It adds one root/view policy stream $S_{\rm policy}(R,V)$ for active root/view pairs with at least 100 reconstructed contexts.  This threshold is not a browser-policy measurement and is not used in the headline comparison; it is a deterministic source-contract variant over the same reconstructed corpus.  It affects 82 of 744 active root/view pairs and 34,417 of 44,912 path/view contexts.  The augmented profile changes the global distribution from $6/16/18$ to $7/17/19$ and, more importantly, produces equal-depth contexts with different dependency counts.

\begin{table}[H]
\caption{Semantic heterogeneity sensitivity.  The core profile is regular in depth; the policy-augmented profile adds a root/view policy stream for selected root/view scopes and therefore splits equal-depth contexts.}
\label{tab:heterogeneity}
\centering
\scriptsize
\begin{tabular}{@{}rrrr@{}}
\toprule
Depth & Core $m$ & Augmented $m$ values & Context counts \\
\midrule
0 & 2 & 2 / 3 & 662 / 82 \\
1 & 4 & 4 / 5 & 6,088 / 5,657 \\
2 & 6 & 6 / 7 & 2,711 / 10,559 \\
3 & 8 & 8 / 9 & 1,033 / 7,497 \\
4 & 10 & 10 / 11 & 102 / 2,949 \\
\bottomrule
\end{tabular}
\end{table}

Table~\ref{tab:heterogeneity} is intentionally scoped as sensitivity evidence, not a claim of complete browser-policy reconstruction.  Its role is to separate the typed compiler mechanism from the regularity of the core profile: richer declared profiles can add or merge dependencies without changing the LR+ warm-wire semantics.

\subsection{Dependency and state distribution}

The reconstructed path/view trace has CA-edge depth distribution
{\scriptsize
\[
\begin{array}{c|rrrrrrrrr}
d&0&1&2&3&4&5&6&7&8\\\hline
N&744&11{,}745&13{,}270&8{,}530&3{,}051&1{,}912&2{,}100&2{,}056&1{,}504
\end{array}
\]
}
Applying $m(d)=2+2d$ gives $m\in\{2,4,6,8,10,12,14,16,18\}$ with median 6, p95 16, and maximum 18.  The evaluated core state contains 744 root-profile records, 744 root-authorization records, 35,266 selected-edge records, and 35,266 selected-context status records, for 72,020 typed dependency records.  The profile/root equality is a canonicalization fact, not an unexplained artifact: the audit finds exactly 744 active $(R,V)$ keys, four view-only keys, and no unused profile records.  Each root-profile record $S_{\rm rprof}(R,V)$ is paired with the corresponding root-authorization stream $S_{\rm root}(R,V)$ for the same active root/view key.  The count also matches the 744 depth-zero root contexts because those rows are precisely the active selected-root/view contexts.  The status/edge equality is likewise expected because the evaluated status stream is scoped to the selected transition $(X_i,X_{i+1},V)$, while optional CA/view-global status deduplication is not used in the headline accounting.  The all-view retained-state serialization is 16,150,580 bytes, and the signed full checkpoint is 16,160,724 bytes.

\subsection{Update and processing boundary}

A signed delta changing $k$ typed records is modeled as
\[
B_{\Delta}(k)=2,584+264k.
\]
The one-record signed delta is 2,848 bytes.  Record-only updates change $R_e$ and the update-plane state identifier, but do not rewrite stable CA-context terminals.  Because the warm selector is state-local, unrelated record updates do not require origins to retain old global states.  They only change the client's installed state through the background update plane.

The artifact verifies the protocol and reconstruction evaluation pipeline; it does not include a full browser/TLS benchmark.  ML-DSA verification occurs when state is installed.  Warm processing consists of parsing the selector, deriving the selected terminal, checking the current compiler vector, resolving the selected records, and evaluating source horizons.

\paragraph{Evaluation boundary.}
The evaluation reports global and per-view typed-compiler communication, retained-state, checkpoint, and lifecycle claims over the pinned CCADB reconstruction.  It does not include live OCSP/CRLite measurements, browser path-building behavior, or TLS integration.  These limits are explicit rather than hidden in the byte tables.

\section{Lifecycle and Deployment}
\label{sec:lifecycle}

Warm-message bytes alone do not determine deployment cost.  LR+ shifts work into retained state and background updates, while SPB keeps clients stateless and carries a larger object on every warm authentication.  The daily online model charges one signed update batch for the day:
\[
C_{\rm LR+}(N,k,H)=N B_{\rm LR,warm}+B_{\Delta}(k)+\frac{B_{\rm ckpt}}{H},
\]
where $N$ is warm authentications/day, $k$ is changed typed records/day, and $H$ is checkpoint amortization in days.  This is also the cost of a one-day range delta, so the daily crossover is unchanged by adding range catch-up support.  SPB costs
\[
C_{\rm SPB}(N)=N B_{\rm SPB,warm}.
\]
Using each view's median LR+ and SPB warm sizes, its own checkpoint size, $H=30$, and $k=16$, LR+ crosses SPB at 19.60 warm authentications/day for Chrome, 37.98 for Mozilla, 47.02 for Apple, and 50.41 for Microsoft.

\begin{table}[H]
\caption{Lifecycle model by view with 30-day checkpoint amortization and 16 changed typed records/day.}
\label{tab:lifecycle}
\centering
\small
\begin{tabular}{@{}lrrr@{}}
\toprule
View & LR+ med. & SPB med. & Crossover/day \\
\midrule
Apple & 296 B & 3,842 B & 47.02 \\
Chrome & 296 B & 3,842 B & 19.60 \\
Microsoft & 392 B & 4,260 B & 50.41 \\
Mozilla & 296 B & 3,842 B & 37.98 \\
\bottomrule
\end{tabular}
\end{table}

\paragraph{Sensitivity.}
Because $k=16$ and $H=30$ are operating points rather than measurements of all deployments, the artifact also evaluates a crossover surface for $k\in\{1,4,16,64,256\}$ and $H\in\{7,30,90\}$.  Table~\ref{tab:lifecycle-sens} gives representative values.  Larger checkpoint amortization horizons lower the crossover, while high update rates increase it.  The retained-state conclusion is therefore strongest for high-reuse clients and longer amortization windows, and weakest for fresh or rarely used clients.

\begin{table}[H]
\caption{Lifecycle crossover sensitivity in warm authentications/day.  Entries use median per-view LR+ and SPB bytes.}
\label{tab:lifecycle-sens}
\centering
\scriptsize
\begin{tabular}{@{}llrrrrr@{}}
\toprule
View & $H$ & $k=1$ & $k=4$ & $k=16$ & $k=64$ & $k=256$ \\
\midrule
Chrome & 7 & 76.58 & 76.80 & 77.70 & 81.27 & 95.56 \\
Chrome & 30 & 18.48 & 18.71 & 19.60 & 23.17 & 37.47 \\
Chrome & 90 & 6.70 & 6.92 & 7.81 & 11.39 & 25.68 \\
Microsoft & 7 & 209.25 & 209.46 & 210.27 & 213.55 & 226.66 \\
Microsoft & 30 & 49.39 & 49.59 & 50.41 & 53.69 & 66.79 \\
Microsoft & 90 & 16.95 & 17.16 & 17.98 & 21.25 & 34.36 \\
\bottomrule
\end{tabular}
\end{table}

The conclusion is not that LR+ universally dominates SPB.  LR+ targets high-reuse relying parties such as browsers and enterprise clients that amortize retained state across many authentications.  Very low-use fresh clients may reasonably prefer SPB.  The important locality result is different: unrelated updates still cost background bytes, but they do not force origins to support the client's exact old global state identifier and do not block selected warm authentication.

Offline clients either catch up by consecutive deltas, use one signed range catch-up delta, or reinstall a checkpoint.  With 16 changed records/day, one consecutive daily delta is $B_\Delta(16)=6{,}808$ bytes, so seven retained daily transitions cost $47{,}656$ bytes.  A signed range catch-up delta from the client's installed endpoint to the seven-day endpoint carries the 112 changed record payload once and costs $B_\Delta(112)=32{,}152$ bytes; this is the value used in the stale-client scenario.  If the client is outside the delta-retention window, it reinstalls the applicable view checkpoint from Table~\ref{tab:per-view}.  In all cases warm acceptance remains bounded by selected dependency horizons; a stale client beyond the relevant horizon fails closed until state is refreshed.

\section{Related Work and Scope}
\label{sec:related}

\paragraph{Certificate and signature compression.}
TLS certificate compression reduces transport size by applying generic compression to certificate chains~\cite{rfc8879}.  MTCs and MTL-style amortization move toward authenticated state and aggregate or landmark-relative certificate representations~\cite{mtc05,mtl09}.  MTC is richer than an issuance-proof compression format: the current draft includes relying-party maintained issuance-log state such as trusted subtrees and revoked serial ranges.  LR+ is complementary rather than a replacement.  It asks which mutable selected-authorization facts must be authenticated when a compressed representation is reused in a selected context, including root/profile state, selected-parent authorization, status or distrust, policy generation, and bounded grants.

\paragraph{Validation policy and root-store semantics.}
Hammurabi separates X.509 validation policy from mechanism with a pluggable logic-based framework~\cite{hammurabi}.  No Root Store Left Behind proposes General Certificate Constraints for fine-grained root trust and partial distrust across root stores~\cite{noroot}.  These works are closely aligned with LR+'s declared-profile boundary: they make policy semantics or root trust more explicit, while LR+ asks which mutable policy and authority facts must be authenticated and retained so that compressed authentication preserves the selected decision over time.

\paragraph{Revocation and WebPKI state.}
OCSP and CRLs define status mechanisms with their own semantics and caching rules~\cite{rfc6960,rfc5019,rfc5280}.  CRLite pushes compressed revocation information to browsers and illustrates the need for explicit coverage statements~\cite{crlite17,crlite26}.  RITM, PKISN, and broader delegation/revocation studies show that revocation and root-program state are ecosystem-level problems rather than simple signature checks~\cite{ritm,pkisn,sokdr}.  LR+ does not replace those sources; it gives a way to bind their effective values and horizons to selected CA contexts.

\paragraph{Interface separation.}
Context-separable interfaces study when cryptographic interfaces remain safe under key reuse across contexts~\cite{contextsep}.  LR+ uses a different notion of context: not secret-key interface separation, but mutable selected-authorization state surrounding an authentic certificate object.  The conceptual overlap is that both lines of work make an implicit context boundary explicit before proving security.

\paragraph{Authenticated data structures and transparency.}
Merkle trees, authenticated dictionaries, and dynamic authenticated indexes provide membership, freshness, and completeness tools~\cite{li06authdb}.  Certificate Transparency authenticates append-only log views~\cite{rfc9162}.  LR+ uses standard authenticated-state ideas but focuses on a different question: whether the authenticated projection is sufficient for the selected authorization semantics under independently mutable authority context.

\paragraph{Cross-signing and PQ migration.}
Cross-signing creates multiple admissible lineages for the same subject/key identity, and public measurements show it is a common WebPKI practice~\cite{crosssign20}.  Post-quantum migration increases the value of amortizing authentication, as also seen in adjacent PKI settings such as pqRPKI~\cite{pqrpki}.  The novelty of LR+ is not caching or Merkle authentication alone; it is the context-closure condition and its realization with path-local horizons, CA-context terminals, and state-local warm selectors.

\section{Discussion and Limitations}
\label{sec:discussion}

LR+ is a semantic compression layer, not a complete browser implementation.  The artifact measures retained-state serialization and a lifecycle model; it does not include full TLS integration, live browser telemetry, or live OCSP/CRLite acquisition.  Source contracts are explicit assumptions or deployment obligations.  The update publisher is assumed pinned and policy-consistent; equivocation detection can be added through transparency or witnessing but is not part of the core theorem.

The evaluation uses a pinned CCADB reconstruction rather than an aggregate-only evidence.  It reports global typed-compiler values, per-view quantiles, per-view state/checkpoint sizes, and per-view lifecycle crossovers.  The remaining engineering work before deployment is to instantiate source adapters, integrate with browser update channels, and decide which leaf-local status semantics remain in the baseline.

\section{Conclusion}

Post-quantum WebPKI compression must preserve more than object authenticity.  A relying party accepts a selected authorization context, and that context depends on mutable authority state.  LR+ makes this dependency explicit: F/C/L witnesses close event ambiguity, typed CA-context dependencies define what is retained, and state-local warm selectors let repeated authentications reuse authenticated state without global warm-state synchronization.  The pinned CCADB reconstruction shows that the stronger semantics still give a large warm-bandwidth advantage over a one-signature stateless baseline for high-reuse relying parties, while exposing the retained-state and update-plane costs needed to obtain that advantage.

\appendix
\section{Open Science and Artifact}
\label{app:open-science}

The implementation and evaluation source is accessible on GitHub at \url{https://github.com/nserser/LR-WebPKI}.  The README provided in that repository describes all steps required to reproduce the results reported in this paper.

\section{Additional Event-Coverage Details}
\label{app:event-coverage}

The proof of Theorem~\ref{thm:fcl} relies on two restrictions that are easy to miss.  First, dependency completeness is profile-relative: the compiler is responsible for emitting every mutable effective value consumed by the declared validation profile.  Second, event toggles must be globally realizable, not merely locally imaginable.  The sufficient condition in Lemma~\ref{lem:toggle} is therefore useful in practice: correlated facts are grouped into composite records and independently declared streams must satisfy a conditional rectangularity condition after the baseline, pre-state, and other histories are fixed.

The F/C/L terminology should be read semantically.  F determines the effective branch, C excludes a hidden branch by complete coverage, and L neutralizes an event by making policy insensitive to it before expiry.  A protocol field that looks like a freshness timestamp is not automatically an F witness; a filter digest is not automatically a C witness; and an expiry time is not automatically an L witness if early revocation is possible.

\section{LR+ Algorithms}
\label{app:lrplus-algorithms}

This appendix records the concrete state-machine checks used by the reference artifact.  The notation is intentionally close to the implementation; the main text gives the security abstraction.

\paragraph{Typed compilation.}
Given $\bar P=(X_0,\ldots,R,V)$, the compiler emits root-profile and root-authorization records, plus selected-edge and selected-context status records for each CA edge.  Each emitted record carries a semantic type, stable stream identifier, context identifier, generation, effective value digest, source-contract mode, horizon, evidence digest, and inclusion reason.  The compiler orders records canonically as root-profile, root-authorization, then edge/context-status pairs from the selected issuer toward the root.  The ordered list is the dependency vector hashed into the terminal.

\paragraph{Record identity.}
A concrete certificate identifier is $H(\textsf{CERT-ID},\operatorname{DER}(X))$, so cross-signed certificates with the same Subject-SPKI remain distinct.  A path identifier is $H(\textsf{PATH-ID},\operatorname{enc}(V,\mathsf{cid}(X_0),\ldots,\mathsf{cid}(R)))$.  A stream identifier is derived from the semantic type and stable authority context, not from the current epoch, and a stable reference is $H(\textsf{REF},\operatorname{enc}(\mathsf{strid}(s)))$.  Stable references remain usable across unrelated record updates.  A generation changes when the effective value for that stream changes.  The evidence digest commits to the source witness or profile-specific evidence object used by the publisher.

\paragraph{Consecutive delta installation.}
A consecutive delta is accepted only if all of the following hold: the configuration and update key are pinned; the predecessor state identifier equals the installed state; the epoch is consecutive; the signature verifies; the mutable and topology deltas are canonical; applying the deltas to local maps recomputes the signed successor roots; and the successor state identifier equals the signed value.  If there is no topology delta, the topology version and root remain unchanged.  If the topology root changes, the topology version must increase.

\paragraph{Range catch-up delta installation.}
A range delta is domain-separated from consecutive deltas.  It names an exact installed predecessor $(e,v,\mathsf{sid}_{e,v})$ and a later endpoint $(e',v',\mathsf{sid}_{e',v'})$, carries final values for the records changed across the range, and is signed by the pinned publisher.  The verifier rejects unless $e'>e$, the predecessor equals the installed state, the signature verifies, every patched record generation is greater than its installed generation, applying the patches recomputes the signed successor mutable root and state identifier, and topology root/version monotonicity holds.  Records not carried in the range payload are interpreted as unchanged from the installed predecessor.  Accepted range installation is therefore endpoint-equivalent to installing the corresponding valid consecutive deltas, but uses one publisher signature and one fixed envelope.

\paragraph{Checkpoint installation.}
A checkpoint carries the complete mutable map and terminal topology, their roots, the epoch and topology version, and a signed digest of the snapshot.  A client accepts it only if the signature verifies, the snapshot recomputes the signed roots, the epoch is newer than the installed epoch, and topology monotonicity is respected.  Checkpoints are not freshness guarantees for selected paths; after installation, warm acceptance is still bounded by the selected dependency horizons.

\paragraph{Warm verification.}
A warm wire selector contains no global $\mathsf{sid}$, terminal key, terminal check value, or authority value.  It contains only the 8-byte header $\tau=(\mathsf{type}:1,\mathsf{version}:1,\mathsf{flags}:2,m:4)$ and a canonical vector of 48-byte stable references.  The parser rejects unsupported types or versions, nonzero reserved flags, out-of-range counts, truncated encodings, and trailing bytes.  The verifier validates baseline object evidence, obtains the exact selected CA-context lineage, derives $\mathsf{pid}_{\bar P}$, $d_{\rm wire}$, and $c_{\rm wire}$ from that lineage and reference vector, and checks that installed $T_v[\mathsf{pid}_{\bar P}]$ equals $(m,d_{\rm wire},c_{\rm wire})$.  The terminal invariant makes this equality a check against the current compiler output, not merely against an old authenticated terminal.  The verifier then resolves each reference in installed $M_e$, checks source-contract mode and horizon, and evaluates the typed effective values under the declared profile.  The selector is only a pointer into installed authenticated state; it is not a self-contained proof of authority freshness.

\paragraph{View/path negotiation.}
In profile-hint mode the verifier or relying-party stack offers a coarse closure-profile identifier.  The origin maps that identifier and its available certification paths to one candidate warm selector.  In catalog mode, the origin sends a bounded list of pairs \((\mathsf{profid}_j,\mathsf{WireWarm}_j)\).  The client selects exactly one matching profile identifier before warm verification.  A selector for the wrong trust view, wrong selected parent chain, or unsupported profile fails because the locally derived \(\mathsf{pid}_{\bar P}\) does not match the current topology entry.  The profile identifier never contains the client's installed epoch, state root, topology root, or global \(\mathsf{sid}\).

\paragraph{Failure cases.}
A wrong trust view changes the CA-context lineage and fails terminal lookup.  A reordered, duplicate, old, or wrong-count dependency vector fails the current-terminal check unless it equals the current compiler output.  A range delta from the wrong predecessor, with a missing changed record, non-monotone generation, or mismatching endpoint root fails update-plane installation.  A selected-record reduction changes the effective value and fails policy or horizon checks.  An unrelated record update may change $R_e$ and $\mathsf{sid}_{e,v}$, but the same warm selector remains valid when interpreted against the new installed state, provided its selected records still satisfy policy and horizon checks.

\section{Adaptive Security Game}
\label{app:security-game}

The adversary receives public parameters and can request honest state transitions, checkpoints, and warm selectors for supported CA contexts.  It may interleave these queries adaptively and choose the final accepted transcript after seeing previous updates.  It wins if a verifier accepts a warm transcript for a full validation instance whose ideal selected semantics reject.

\paragraph{Bad events.}
The game exposes five modeled failure classes.  $\mathsf{SigBad}$ is a valid update or checkpoint not signed by the pinned publisher.  $\mathsf{HashBad}$ is a collision or ambiguous canonical encoding that makes two states, records, stream identifiers, path identifiers, ordered-reference vectors, or terminal statements share the same digest.  $\mathsf{BaseBad}$ is an ordinary baseline-validation failure for the selected end-entity evidence.  $\mathsf{ClockBad}$ is a violation of the trusted-time rule.  $\mathsf{SrcBad}_i$ is failure of the source contract for a selected dependency $S_i$.  The selected-source event is adaptive: after the adversary chooses its first accepted winning transcript, the game charges $\exists S_i\in\Dep(P^\star):\mathsf{SrcBad}_i$ for that selected context $P^\star$.

\paragraph{Hybrid sequence.}
Game $G_0$ is the real experiment.  Game $G_1$ aborts on update/checkpoint statements that verify without being honestly produced.  A difference between $G_0$ and $G_1$ yields a PQ-EUF-CMA adversary against ML-DSA.  Game $G_2$ aborts on hash or canonical-encoding collisions in state roots, terminal records, path identifiers, stream identifiers, ordered references, or evidence commitments.  A difference yields a quantum collision-resistance adversary.  Game $G_3$ aborts on baseline or clock failure.  In the remaining game the installed state is an honestly authenticated publisher state.

\paragraph{State-local warm argument.}
Because the warm wire selector has no global state identifier, the proof does not need to show that the origin and client named the same epoch.  It shows instead that the verifier interprets every selector against its own installed authenticated state.  Terminal fields are not adversarial wire fields: $\mathsf{pid}_{\bar P}$, $d_{\rm wire}$, and $c_{\rm wire}$ are derived from the selected CA-context lineage and transmitted reference vector before lookup.  The topology is keyed by $\mathsf{pid}_{\bar P}$ and contains exactly the current compiler vector.  Therefore accepted warm verification implies that the wire references equal $\mathsf{Refs}_v(\bar P)$ except under $\mathsf{HashBad}$.  Each current reference resolves to an authenticated typed record in $M_e$.  If none of the selected source contracts fails and time is within the selected horizon, the effective values equal the ideal values used by the declared profile.

\paragraph{Conclusion.}
Full-path lifting binds the baseline end-entity evidence to the exact selected CA-context lineage.  Compiler completeness fixes all mutable CA-context decisions.  Therefore the accepted transcript has the same selected semantics as the ideal profile, contradicting the adversary's win condition.  The remaining probability is bounded by the signature, hash, baseline, clock, and adaptive selected-source term in Theorem~\ref{thm:security}.  Corollary~\ref{cor:adaptive-src} gives the simpler path-local summation only when the source contracts themselves hold under conditioning on the prior adversarial transcript and on final selected-dependency membership.

\end{document}